\documentclass[twoside,leqno,twocolumn]{article}

\usepackage[letterpaper]{geometry}

\usepackage{siamproceedings}

\usepackage[T1]{fontenc}
\usepackage{amsfonts}
\usepackage{graphicx}
\usepackage{epstopdf}
\usepackage{enumitem}
\usepackage{algorithmic}
\usepackage{authblk}

\ifpdf
  \DeclareGraphicsExtensions{.eps,.pdf,.png,.jpg}
\else
  \DeclareGraphicsExtensions{.eps}
\fi

\newsiamremark{remark}{Remark}
\newsiamremark{hypothesis}{Hypothesis}
\crefname{hypothesis}{Hypothesis}{Hypotheses}
\newsiamthm{claim}{Claim}

\usepackage{amsopn}

\newcommand{\fold}{\text{fold}}
\newcommand{\ffold}{\text{ffold}}

\begin{document}

\newcommand\relatedversion{}
\renewcommand\relatedversion{\thanks{The full version of the paper can be accessed at \protect\url{https://arxiv.org/abs/0000.00000}}} 

\title{\Large Improving Pinwheel Density Bounds for Small Minimums}
\author[]{Ahan Mishra}
\author[]{Parker Rho}
\author[]{Robert Kleinberg}
\affil[]{Cornell University}

\date{}

\maketitle

\begin{abstract} 
\mbox{}
\vspace{5pt}

The density bound for schedulability for general pinwheel instances is $\frac{5}{6}$, but density bounds better than $\frac{5}{6}$ can be shown for cases in which the minimum element $m$ of the instance is large. Several recent works have studied the question of the ``density gap'' as a function of $m$ \cite{regular} \cite{perpetual} \cite{covering}, with best known lower and upper bounds of $O \left( \frac{1}{m} \right)$ and $O \left( \frac{1}{\sqrt{m}} \right)$. We prove a density bound of $0.84$ for $m = 4$, the first $m$ for which a bound strictly better than $\frac{5}{6} = 0.8\overline{3}$ can be proven. In doing so, we develop new techniques, particularly a fast heuristic-based pinwheel solver and an unfolding operation. 

\vspace{5pt}
\end{abstract}

\section{Introduction}
\mbox{}

\subsection{Problem Statement}
\mbox{}
\vspace{5pt}

The pinwheel scheduling problem asks, given a list of integers $(a_1, a_2, a_3, \dots, a_n)$, interpreted as job periods, whether there exists a map $f : \mathbb{N} \rightarrow [n]$, interpreted as an assignment of which job to perform on each day, such that for every $i$, job $i$ is performed as least once in every interval of $a_i$ days. 

The density of an instance is $D(A) := \sum_{i = 1}^n \frac{1}{a_i}$. 

We can see that if an instance has density greater than 1, it is not schedulable. It has recently been proved that every instance with density at most $\frac{5}{6}$ is schedulable \cite{kawamura}. The corresponding statement is not true for any $d > \frac{5}{6}$ due to $[2, 3, x]$ being unschedulable for all $x$. 

While a $\frac{5}{6}$ density bound is the best that can be said for general pinwheel instances, we study a related problem of finding density bounds given instance minimums. Specifically, let us define $G(m)$ as the largest real number such that every pinwheel instance with all periods at least $m$ and density at most $G(m)$ is schedulable. 

We can see that $G$ is monotonically increasing and due to the $\frac{5}{6}$ result we know that $G(1) = \frac{5}{6}$. 

\vspace{5pt}

\subsection{Related Work}
\mbox{}
\vspace{5pt}

Understanding the density gap (i.e., $1 - G(m)$) as $m$ scales is important because, in practice, we may be interested in scheduling pinwheel instances with a minimum job period significantly greater than 1. 

For example, the original application mentioned in the paper introducing the pinwheel problem is that of a ground station communicating with multiple satellites \cite{original}. The worst-case instance mentioned previously, $[2, 3, x]$, corresponds to having a satellite that must be serviced at least once every two time steps. A more realistic situation might be for a ground station to communicate with dozens of satellites, each of which only needs to be serviced less than 10\% of the time (so job periods greater than 10). 

The current best known upper and lower bound for the density gap are $O \left( \frac{1}{\sqrt{m}} \right)$ and $O \left( \frac{1}{m} \right)$. 

In particular, it is known that for every $m$, any pinwheel instance with elements at most $m$ and density at most $1 - \frac{1 + \ln 2}{1 + \sqrt{m}} - \frac{3}{2m}$ is schedulable \cite{covering}. On the other side, it can be shown that $(m, m + 1, m + 1, \dots, m + 1, x)$ (where there are $m - 1$ repetitions of the period $m + 1$) is unschedulable for any $m$ and $x$ \cite{covering}. 

Therefore, we have
$$1 - \left( \frac{1 + \ln 2}{1 + \sqrt{m}} + \frac{3}{2m} \right) \leq G(m) \leq 1 - \left( \frac{m - 1}{m(m + 1)} \right)$$

Both the upper and lower bounds go to 1 asymptotically, but we may hope for significantly tighter bounds in practice. 

In particular, for any $m \leq 100$, $\frac{1 + \ln 2}{1 + \sqrt{m}} - \frac{3}{2m} \geq \frac{1}{6}$.

In other words, in order to improve on the $\frac{5}{6}$ result which is known via other methods, we need to restrict our attention to instances with period minimums of at least 100. A minimum of 100 is significantly larger than we may hope for in practice. In addition, it has been conjectured that the upper bound is asymptotically tight \cite{covering}, so the exact density bound for instances with period minimum of 100 could be as much as $1 - \frac{1}{100} = \frac{99}{100}$. 

So, there is a significant divergence between theory and practice for minimums less than 100

The only value of $G(m)$ which has been explicitly found exactly is $G(1) = \frac{5}{6}$, but we can easily see that $G(2) = \frac{5}{6}$ and $G(3) = \frac{5}{6}$ as well since $[2, 3, x]$ and $[3, 4, 4, x]$ are unschedulable for all $x$.

\vspace{5pt}

\subsection{Contribution}
\mbox{}
\vspace{5pt}

Since $G(3) = \frac{5}{6}$, the first possible $m$ for which we may be able to prove a density bound greater than $\frac{5}{6}$ is $m = 4$, which we do. Our improved density bound relies on computational assistance and is extensible. We also detail new techniques, particularly an unfolding operation and a fast pinwheel solver. In particular, we 

\begin{enumerate}
    \item Show that $G(4) \geq 0.84$ by proving that every instance with elements at least four and density at most 0.84 is schedulable. 
    \item Implement an unfolding operation which generalizes Kawamura's method for proving the $\frac{5}{6}$ bound \cite{kawamura}, and use it to prove the 0.84 density bound. 
    \item Develop a heuristic-based pinwheel solver which, in some cases, can find schedules much faster than the current best solver, particularly when there are at least 10 jobs. 
\end{enumerate}

\section{An Improved Density Bound for \texorpdfstring{\boldmath$m = 4$}{}}
\mbox{}

\subsection{Proof}\label{subsection:proof}
\mbox{}

\vspace{5pt}

We use the same fold operation introduced by Kawamura \cite{kawamura} and several relevant facts. 

\begin{algorithm}[ht]
  \DontPrintSemicolon
  \caption{Fold Operation $\fold_{\theta}(A)$.}   \label{algo:fold}
  As input, take a pinwheel instance $A = (a_1, a_2, \dots, a_n)$ \;
  
  \While{$\max(A) > \theta$}{
    Let $a$ and $b$ be the largest and second largest values in $A$, allowing for repetition.

    Delete $a$ from $A$.
    
    \If{$b > \theta$}{        
    Delete $b$ from $A$, and add $b/2$ to $A$. 
    }
    
    \Else{
    Add $\theta$ to $A$.
    }
    
    \Return $A$
}
\end{algorithm}

\begin{fact}\label{fact:monotonicity}
(Monotonicity Principle) For any $a \leq b$, if the instance $A \sqcup (a)$ is schedulable, then the instance $A \sqcup (b)$ is schedulable. \cite{kawamura}
\end{fact}


\begin{fact}\label{fact:schedulable}
    If $A$ is unschedulable, then $\fold_{\theta}(A)$ is unschedulable \cite{kawamura}.
\end{fact}

\begin{fact}\label{fact:theta}
    For every $a$ in $\fold_{\theta}(A)$, $a \leq \theta$ \cite{kawamura}.
\end{fact}

\begin{fact}\label{fact:density}
    For any $\theta$, $D(\fold_{\theta}(A)) \leq D(A) + \frac{1}{\theta}$ \cite{kawamura}.
\end{fact}

\begin{lemma}\label{lemma:foldsched}
For any pinwheel instance $A$, there exists $\theta \in \mathbb{N}$ such that $\fold_{\theta}(A)$ is schedulable. 
\end{lemma}

\begin{proof}

This proof is accomplished by enumerating lists across different values of $\theta$. In particular, for every even $\theta$ between 12 and 30, we have a lists.csv and a removed.csv file. Let us denote the corresponding sets of instances as $L_{\theta}$ and $R_{\theta}$, respectively. 

Let us additionally define 
\[
D'_{\theta}(A) = \sum_{i = 1}^n
\begin{cases} 
1/a_i & \text{if } a_i < \frac{\theta}{2}, \\
1/(a_i + 1) & \text{if } a_1 \geq \frac{\theta}{2}.
\end{cases}
\]

and let us define $\ffold_{\theta}(A)$ as the result of taking every element $a$ in $\fold_{\theta}(A)$ and replacing it with $\lceil a \rceil - 1$ (this is the same operation described by Kawamura \cite{kawamura}).

Through explicit enumeration and certification, we know the following properties:

\begin{enumerate}
    \item For $\theta = 12$, every possible integer instance $A$ with elements at least 4, at most $\theta - 1$, and $D'_{\theta}(A)$ at most $0.84 + \frac{1}{\theta}$ is contained either in $L_{\theta}$ or $R_{\theta}$.
    \item For $\theta \in [14, 16, 18, \dots, 30]$, every possible integer instance $A$ with elements at most $\theta - 1$, $D'_{\theta}(A)$ at most $0.84 + \frac{1}{\theta}$, and satisfying $\fold_{\theta - 2}(A) \in R_{\theta - 2}$ is contained either in $L_{\theta}$ or $R_{\theta}$.
    \item For every instance $A$ in $L_{\theta}$, $A$ is schedulable.
    \item $R_{30}$ is empty.
\end{enumerate}

We now provide an algorithm to construct the $\theta$ value guaranteed in the lemma.

\begin{algorithm}[ht]
  \DontPrintSemicolon
  \caption{$\theta$ Generator.}   \label{algo:theta}
  As input, take a pinwheel instance $A = (a_1, a_2, \dots, a_n)$ \;

  Set $\theta \gets 12$
  
  \While{$\theta \leq 30$}{    
    \If{$\textnormal{ffold}_{\theta}(A) \in L_{\theta}$}{
        Return $\theta$
    }

    $\theta \gets \theta + 2$
}
\end{algorithm}

\begin{figure*}[t]
    \centering
    \includegraphics[scale=0.6]{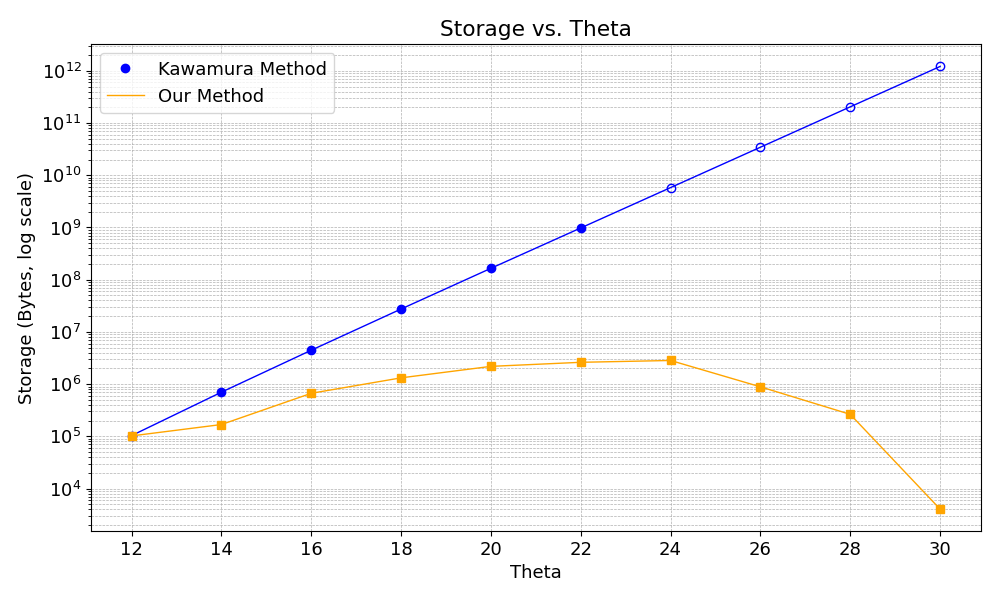}
    \caption{Storage requirements for our method vs. Kawamura's original method \cite{kawamura}.}
    \label{fig:storage}
\end{figure*}

We first prove that \cref{algo:theta} returns an output for every possible pinwheel instance $A$ satisfying $D(A) \leq 0.84$.

Suppose, for the sake of contradiction, that the algorithm returns nothing. 

First of all, note that for any $A$ and any $\theta$, $\fold_{\theta}(A)$ has elements at most $\theta$ and density at most $D(A) + \frac{1}{\theta}$. To go from $\fold_{\theta}(A)$ to $\ffold_{\theta}(A)$ we decrease every changed number by at most 1, so the maximum value becomes $\theta - 1$, and the only non-integers produced by $\fold_{\theta}(A)$ are greater than $\frac{\theta}{2}$, so $D'_{\theta}(\ffold_{\theta}(A)) \leq D(\fold(A))$.

Now, by property 1, for any $A$, $\ffold_{12}(A)$ is either in $L_{12}$ or $R_{12}$. If it were in $L_{12}$, then the algorithm would have returned 12, so it must have been in $R_{12}$.

By property 2, for any $A$ which has $\ffold_{\theta}(A) \in R_{\theta}$, the corresponding value of $\ffold_{\theta + 2}(A)$ is contained in either $L_{\theta + 2}$ or $R_{\theta + 2}$. Since the algorithm returns no output, it cannot be in $L_{\theta + 2}$. Therefore, having $\ffold_{\theta}(A)$ in $R_{\theta}$ implies $\ffold_{\theta + 2}(A)$ being in $R_{\theta + 2}$. 

Inducting, we have that $\ffold_{30}(A)$ is in $R_{30}$. However, $R_{30}$ is empty from property 4, a contradiction. This proves that the algorithm returns some output $\theta_o$.

By property 3, $\ffold_{\theta_o}(A)$ is schedulable, and by the monotonicity principle this means $\fold_{\theta_o}$ is schedulable as well, completing the proof.

\end{proof}

\begin{theorem}
    Every pinwheel instance with elements at least 4 and density at most 0.84 is schedulable.
\end{theorem}

\begin{proof}
    By \cref{lemma:foldsched}, there exists $\theta \in \mathbb{N}$ such that $\text{fold}_{\theta}(A)$ is schedulable. 

    By the contrapositive of \cref{fact:schedulable}, this means that $A$ is schedulable as well, as desired.
\end{proof}

\subsection{Unfolding Operation}
\mbox{}
\vspace{5pt}

Property 2 in \cref{subsection:proof} requires an unfolding operation to implement efficiency. We first describe the core ideas of the unfolding operation that we implemented, then describe the merits of restructuring the proof used by Kawamura \cite{kawamura} into that of \cref{subsection:proof}.

\subsubsection{Implementation}
\mbox{}
\vspace{5pt}

To unfold a large set of lists, we simply unfold each list in the set to get possibilities for the pre-image before folding, and then take the union. So, we describe how to unfold a single list. 

For example, the $[5,7,8,8,11,15,15,15,15]$, which is in $R_{16}$, generates the following two unfolded lists for the $\theta = 18$ case:

\begin{enumerate}
    \item $5, 7, 11, 15, 15, 15, 16, 17, 17, 17, 17$
    \item $5, 7, 11, 15, 15, 15, 17, 17, 17, 17, 17$
\end{enumerate}

Given an instance in $R_{\theta}$ which we want to unfold to have a fold parameter of $\theta + 2$, note that every element that is not $\frac{\theta}{2}$ or $\theta - 1$ must be in the unfolded instance. For any $\theta/2$ element, it could correspond to either a value of $\theta$ or a value of $\theta + 1$ in the unfolded instance. For each of these cases, at most one element of value $\theta - 1$ can optionally be changed, either to $\theta$ or $\theta + 1$. 

After considering all such transformations as described, we apply the $D'_{\theta + 2}(A) \leq 0.84 + \frac{1}{\theta + 2}$ requirement to determine the final unfolded lists.

\subsubsection{Scalability Improvements}
\mbox{}
\vspace{5pt}

With the Kawamura technique, in which we consider every possible integer instance with elements at most $\theta - 1$ and some density bound, the number of instances to consider increases exponentially in $\theta$. This situation gets worse as the density bound we want to prove increases with increasing $m$ (we proved a bound of 0.84 for $m = 4$).

As shown in \cref{fig:storage}, our method demonstrates significantly better scalability than Kawamura's original method \cite{kawamura}. 

Under Kawamura's method, the $\theta = 30$ case would take approximately 1 petabyte to store the set of lists that need to be solved (note that for $\theta \in [24, 26, 28, 30]$) the storage requirements are projected from those for $\theta \in [12, 14, 16, 18, 20, 22]$. Even if storing 1 petabyte were feasible for some, communicating the files would significantly strain networking resources. 

With the current method and compute resources for full surface generation, it would take approximately 4 years (projected) to list all the instances that need to be solved for the $\theta = 30$ case. The time needed to solve all such instances would likely be much longer. These conditions also pose significant issues for tractable  certifiability. 

On the other hand, the total storage required for list enumeration under our method is 10 MB, and 6 seconds are needed for certification. 

\section{A Faster Pinwheel Solver}
\mbox{}

\begin{algorithm}
  \DontPrintSemicolon
  \caption{Partition Filter.}   \label{algo:partition}
  As input, take a list $L$ and possible partition element $E$ \;

  \If{$\textnormal{len}(E) = 2$}{
  \If{$L[E[1]] - L[E[0]] \leq 2$}{
  Return ``valid''
  }

  \Else{
  Return ``invalid''
  }
  }
  
  \If{$\textnormal{len}(E) = 3$}{
  \If{$(L[E[1]] - L[E[0]]) + (L[E[2]] - L[E[0]]) \leq 2$}{
  Return ``valid''
  }

  \Else{
  Return ``invalid''
  }
  }

  \If{$\textnormal{len}(E) = 5$}{
  $sum \gets 0$

  \For{$i \in [0, 1, 2, 3, 4]$}{
  $sum \gets sum + \left( L[E[i]] - 5 \cdot \left\lfloor \frac{L[E[0]]}{5} \right\rfloor \right)$
  }
  
  \If{$sum \leq 3$}{
  Return ``valid''
  }
  \Else{
  Return ``invalid''
  }
  }

  \Else{
  Return ``invalid''
  }
\end{algorithm}

\begin{figure*}[t]
    \centering
    \includegraphics[scale=0.6]{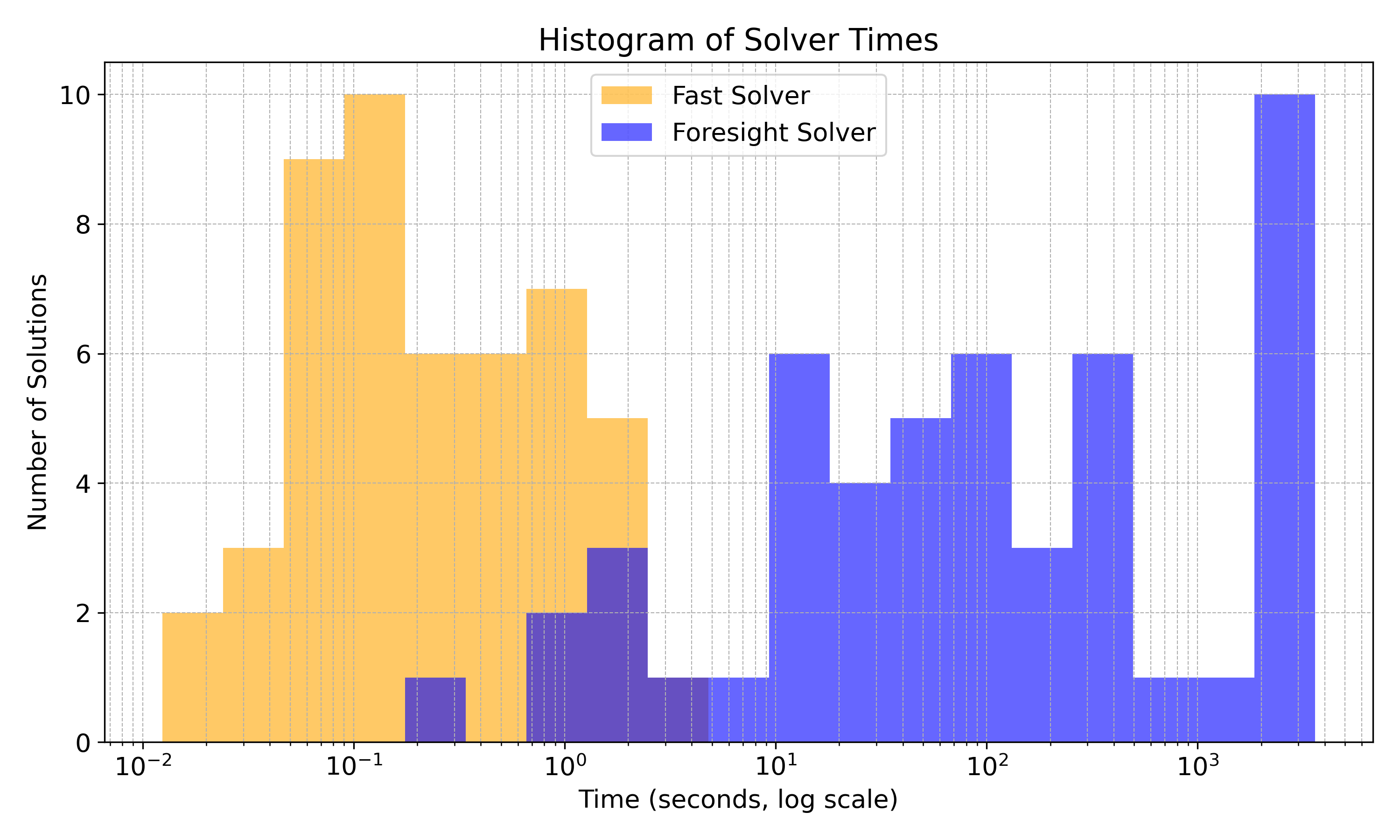}
    \caption{Time required to solve 50 instances using our method vs. the foresight method \cite{towards}.}
    \label{fig:histogram}
\end{figure*}

\begin{figure*}[t]
    \centering
    \includegraphics[scale=0.6]{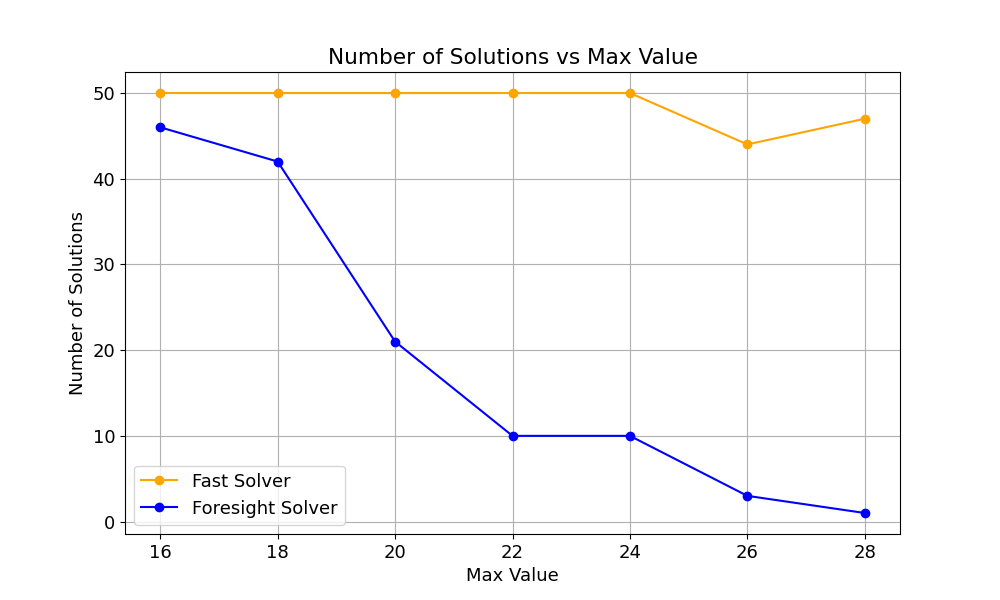}
    \caption{Number of solves from fast vs foresight solvers as max value increases \cite{towards}.}
    \label{fig:solutions}
\end{figure*}

\begin{figure*}[t]
    \centering
    \includegraphics[scale=0.6]{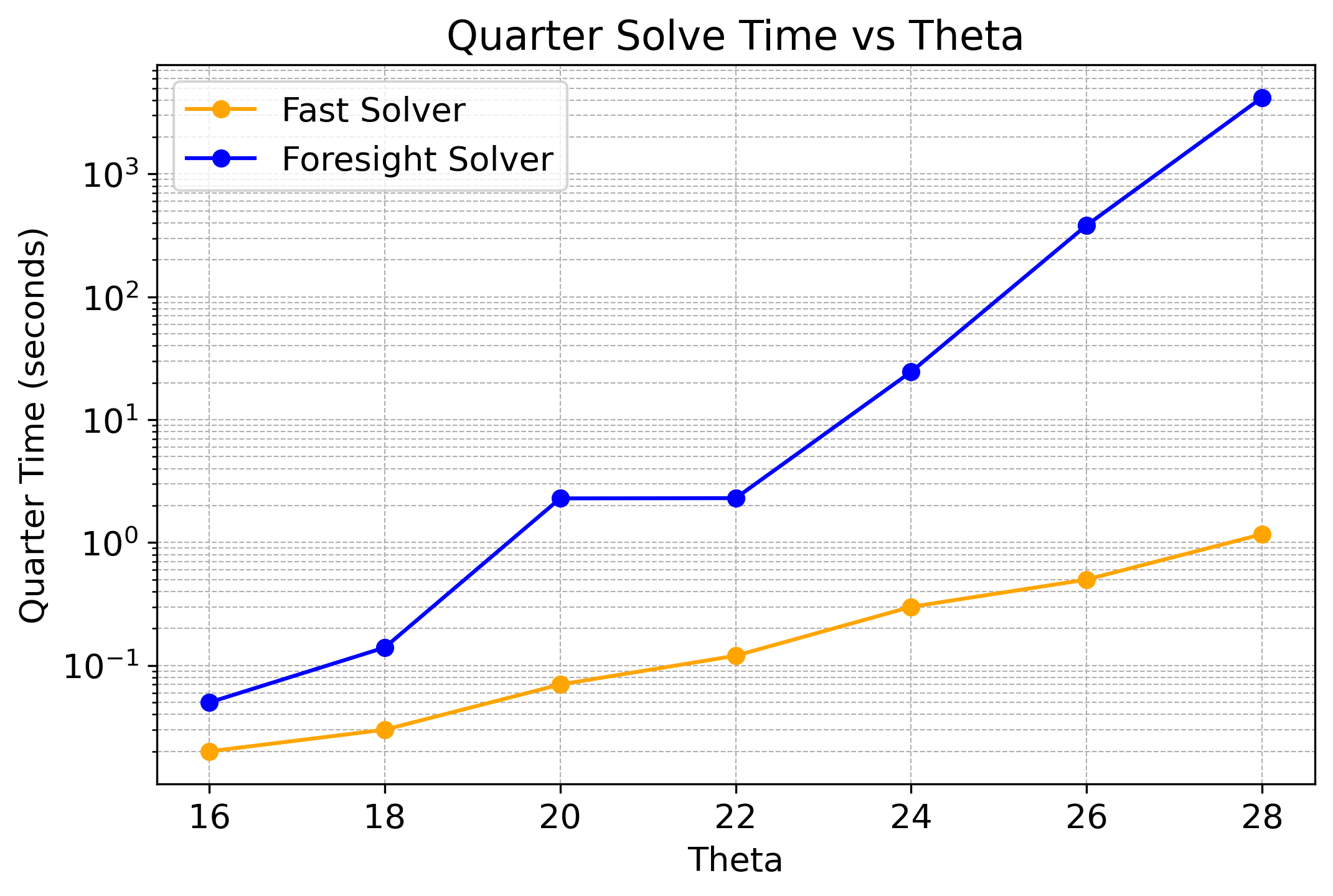}
    \caption{Time needed to solve a quarter of instances as max value increases.}
    \label{fig:quarter}
\end{figure*}

\begin{algorithm}
  \DontPrintSemicolon
  \caption{Fractional Constraint Checker.} \label{algo:checker}
  As input, take a job period with denominator 2 or 3, $j$, the history of scheduling times of the job, $history$, and a next scheduling time, $next$, which we are considering for possible violation \;

  \If{$\textnormal{denominator}(j) = 2$}{
  \If{$\textnormal{len}(history) \geq 2 \textnormal{ and } history[-1] - history[-2] = \textnormal{ceil}(j) \textnormal{ and } next - history[-1] = \textnormal{ceil}(j)$}{
  Return ``constraint violated''
  }
  }
  \If{$\textnormal{denominator}(j) = 3$}{
  \If{$\textnormal{len}(history) < 2$}{
  Return ``constraint passed''
  }
  \If{$\textnormal{numerator}(j) \equiv 1 \pmod{3}$}{
  \If{$\textnormal{len} = 2 \textnormal{ and } history[-1] - history[-2] = \textnormal{ceil} \textnormal{ and } next - history[-1] = \text{ceil}(j)$}{
  Return ``constraint violated''
  }
  \If{$next - history[-1] = \textnormal{ceil}(j) \textnormal{ and } (history[-1] - history[-2] = \textnormal{ceil}(j) \textnormal{ or } history[-2] - history[-3] = \textnormal{ceil}(j))$}{
  Return ``constraint violated''
  }
  }
  \Else{
  \If{$\textnormal{len}(history) > 2 \textnormal{ and } next - history[-1] = \text{ceil}(j) \textnormal{ and } history[-1] - history[-2] = \text{ceil}(j) \textnormal{ and } history[-2] - history[-3] = \text{ceil}(j)$}{
  Return ``constraint violated''
  }
  }
  Return ``constraint passed''
  }
\end{algorithm}

\subsection{Implementation}
\mbox{}
\vspace{5pt}

As a baseline against which to compare the faster pinwheel solver, we use the foresight solver implemented by Gąsieniec et al., which is the fastest published solver, to our knowledge \cite{towards}. This solver, with some enhancement, is also a key component of our fast solver.

The basic idea behind the faster pinwheel solver is that there are many instances for which we can find a schedule that implies a schedule for a particular original instance. In particular, we leverage the partitioning and monotonicity properties \cite{kawamura}. 

For example, the schedulability of the instance [14, 14, 14, 14, 15, 18, 18, 19, 20, 22, 22, 23, 23, 23, 24, 25, 27] is implied by the schedulability of [7, 7, 8, 9, 11, 15, 19, 20, 23, 23, 23]. This is because we have replaced the first two 14's with $14/2 = 7$, the third and fourth 14 with another 7, the two 18's with $18 / 2 = 9$, the two 22's with $22 / 2 = 11$, and the 24, 25, and 27 with $24/3 = 8$.

Theoretically, this means it is ``harder'' to schedule [7, 7, 8, 9, 11, 15, 19, 20, 23, 23, 23]. However, the situation becomes different in the computational setting. In particular, running the foresight solver on the instance [14, 14, 14, 14, 15, 18, 18, 19, 20, 22, 22, 23, 23, 23, 24, 25, 27] takes about 70 minutes to find a solution, whereas running it on [7, 7, 8, 9, 11, 15, 19, 20, 23, 23, 23] takes only 40 seconds. 

The main algorithmic technique for implementing the idea described above more generally is to enumerate a list of good ``partitions'' and implement an ordering mechanism.

Each element of a partition represents a set of elements which are folded into a single element. For example, the partition representing the transition from [14, 14, 14, 14, 15, 18, 18, 19, 20, 22, 22, 23, 23, 23, 24, 25, 27] to [7, 7, 8, 9, 11, 15, 19, 20, 23, 23, 23] is [[0, 1], [2, 3], [5, 6], [9, 10], [14, 15, 16]].

Our algorithm only allows for partition elements with length 2, 3, or 5 (although this can be easily extended), and applies a filter to determine valid elements.

The high-level idea of the filter (\cref{algo:partition}) is to only allow fold operations which do not increase density by too much. Now, whole partitions can be generated recursively such that each element is valid.

The list of possible partitions can be large (1000 or more), and we do not want to consider every partition, so we need a mechanism for ordering partitions based on how likely we think they can generate a quick solution. 

Every partition corresponds to a particular folded list, as we discussed earlier. We analyze two parameters of this list for partition ordering: density and product of elements. In particular, we impose a hard cap on density of 0.95 and consider elements in increasing order of $\sqrt{product} / (0.95 - density)^2$.

The product of elements correlates with the time needed to solve an instance since it is the time it would take to traverse the full state space of an instance \cite{towards}. Density correlates because generally, instances with higher density are more challenging to solve. 

Once we can generate partitions, it is possible to use good partitions and delegate to the foresight solver, and this method already achieves significant performance benefits. To do even better, we instrument the foresight solver with the ability to handle fractional elements, specifically with denominators of 2 and 3. Allowing these denominators is also the reason why only the $\text{len} = 5$ case in \cref{algo:partition} involves a floor operation.

To our knowledge, this is the first solver that allows even a restricted set of non-integer elements. In particular, for a job with real-valued period $r$, in any interval of length $a$, the job must be scheduled at least $\lfloor \frac{a}{r} \rfloor$ times. The key difficulty in allowing for fractional periods is that this condition requires considering an infinite number of possible $a$ (or at least the length of the list currently under consideration). 

Integer periods require us to only look at the previous element to determine whether the next scheduling is valid, but fractional elements require us to look back further. We have designed a casework-based algorithm which can test for a constraint violation in $O(1)$ time in the cases of denominators of 2 and 3 (\cref{algo:checker}). 

The last step of our algorithm is to rearrange the solution of the folded instance into that of the original instance, which we achieve by storing additional mapping metadata during folding and mapping each job of the folded instance into a series of jobs for the original instance.

\vspace{5pt}

\subsection{Results}
\mbox{}
\vspace{5pt}

As a proof of concept for the power of the faster pinwheel solver, we generate 50 random instances with a density of around 0.9 and record the time it takes to solve them using the foresight algorithm compared to our fast algorithm. 

In particular, we use the following rejection sampling process to generate each instance: 

\begin{enumerate}
    \item Generate a uniformly random integer $n$ between 10 and 15.
    \item Generate $n$ independent and uniformly random integers between 1 and 25
    \item Treating the list of $n$ integers as a pinwheel instance, compute the density of the instance (sum of inverses). If the density is between 0.89 and 0.91, return the instance; otherwise, restart the process.
\end{enumerate}

The average time to solve the instances under our fast method was 0.48 seconds, whereas that under the foresight method was 808 seconds, meaning that we have an over 1600x improvement over the foresight solver. Furthermore, both solvers were limited to 1 hour per instance. While the fast solver was able to solve all instances within the time limit, the foresight solver was only able to solve 41 (82\%) of the instances. The complete histogram of results is given in \cref{fig:histogram}.

We now provide results in the direction that the fast solver has significantly better scalability properties than the foresight solver. As with the unfolding operation, this means making the schedulability of large sets of instances more tractable, as is needed for improving density bounds for small $m$, as well as for other problems, such as bamboo garden trimming \cite{perpetual}.

To study scalability of the fast solver, we need a parameterized method to generate random instances. We use the following random process, given a parameter $max$:

\begin{enumerate}
    \item Generate a uniformly random integer $n$ in between $\lfloor \frac{max}{3} \rfloor$ and $\lfloor \frac{2 \cdot max}{3} \rfloor$.
    \item Generate $n$ independent and uniformly random integers between $\lfloor \frac{max}{2} \rfloor$ and $max$.
    \item If the density of the corresponding instance is between 0.89 and 0.91, return the instance, otherwise restart the process.
\end{enumerate}

In \cref{fig:solutions} we show how the number of solutions generated by the foresight method scales compared to that of the fact method. In particular, both solvers are given 10 seconds per instance in each case. The fast method scales well even under a constant time limit. On the other hand, the foresight method rapidly degrades, with only one solution discovered within 10 seconds at a maximum value of 28.

We now study essentially the complement of the problem in \cref{fig:solutions}. Specifically, we previously showed how the number of solutions scale given constant time. In \cref{fig:quarter} we study how the amount of time needed scales given a constant fraction of instances to solve. Specifically, we record the time needed to solve the fastest quarter (specifically 12 out of 50) of instances for each possible maximum value between 16 and 28. 

Both the fast solver and foresight solver demonstrate exponential growth in the time needed to solve the first quartile. However, the fast solver has a smaller exponential growth factor, as seen by the smaller slope of the line in \cref{fig:quarter}. This is critical for solving instances for relatively small values of $\theta$ in folding-related problems, but beyond the $\theta = 22$ value studied by Kawamura \cite{kawamura}. The fast solver goes from a 2.5x performance improvement in the max case of 16 to a 3580x improvement in the max case of 28. 

\section{Conclusion}
\mbox{}
\vspace{5pt}

We have shown that $G(4) \geq 0.84$. Previously, no bound better than $\frac{5}{6}$ was known for any $m$ less than 100, and in fact, four is the smallest $m$ for which a bound better than $\frac{5}{6}$ can be proven. 

Our unfolding operation and fast pinwheel solver demonstrate exponential improvements over the previous best-known methods. We hope that future work can leverage these to establish a strictly increasing sequence of density bounds for $m = 5, 6, 7, 8, 9, \text{ and } 10$. 

Another direction for future work is to implement additional optimizations to the fast solver, particularly enabling it to scale polynomially in restricted cases.

\bibliographystyle{siamplain}
\bibliography{references}

\end{document}